\documentclass[a4paper, 11pt]{article}
\usepackage{fullpage}
\usepackage{booktabs}
\usepackage{amsthm}
\usepackage{thmtools,thm-restate}
\usepackage{microtype}
\usepackage{hyperref}
\usepackage{mathtools,amsmath,xspace,xcolor}

\usepackage{amssymb}
\newtheorem{theorem}{Theorem}
\newtheorem{lemma}[theorem]{Lemma}

\newtheorem{observation}[theorem]{Observation}

 \let\leq\leqslant
\let\ge\geqslant \let\le\leqslant

\newcommand{\arr}{\ensuremath{\mathcal{A}}}
\newcommand{\arrp}{\ensuremath{\mathcal{A}^+}}
\newcommand{\ilev}[1]{\ensuremath{\mathsf{lev}_{#1}}}

\newcommand{\lev}[1]{\ensuremath{\mathsf{lev}(#1)}}
\newcommand{\twin}[1]{\ensuremath{\mathsf{tw}(#1)}}
\newcommand{\cl}{\ensuremath{\ell_\mathsf{cap}}}

\newcommand{\conv}[1]{\ensuremath{\mathsf{CH}(#1)}}

\title{Bichromatic Line-Centers for Point Pairs\thanks{J. Lee, Y. Bae, H.-K. Ahn are partly supported by the National Research Foundation of Korea (NRF) grant funded by the Korea government(MSIT) (RS-2023-00219980), and the Institute of Information \& communications Technology Planning \& Evaluation (IITP) grant funded by the Korea government (MSIT) (No.RS-2019-II191906, Artificial Intelligence Graduate School Program (POSTECH)) and (No. 2017-0-00905, Software Star Lab (Optimal Data Structure and Algorithmic Applications in Dynamic Geometric Environment)).}}
\author{Jaegun Lee\thanks{These authors contributed equally.}~\thanks{Department of Computer Science and Engineering, Pohang University of Science and Technology, Pohang, Republic of Korea. 
\texttt{jagunlee@postech.ac.kr}}
 \and Youjung Bae\footnotemark[1]~\thanks{Graduate School of Artificial Intelligence, Pohang University of Science and Technology, Pohang, Republic of Korea. \texttt{youjungbae@postech.ac.kr}} \and Taehoon Ahn\thanks{Department of Computer Science, Sookmyung Women's University, Seoul, Republic of Korea. \texttt{taehoon@sookmyung.ac.kr} } \and Sang Won Bae\thanks{Division of AI Computer Science and Engineering, Kyonggi University, Suwon, Republic of Korea. \texttt{swbae@kgu.ac.kr} } \and Hee-Kap Ahn\thanks{Department of Computer Science and Engineering, Graduate School of Artificial Intelligence, Pohang University of Science and Technology, Pohang, Republic of Korea. \texttt{heekap@postech.ac.kr}}}

\begin{document}
\date{}
\maketitle

\begin{abstract}
We study the \emph{bichromatic line-center problem} for $n$ pairs of points in the plane.
A feasible solution assigns one point from each pair to the red set $R$ and the other to the blue set $B$.
The goal is to minimize $\max\{w^\circ(R),\,w^\circ(B)\}$, where $w^\circ(X)$ denotes the minimum width of a strip enclosing $X$; the midlines of the corresponding optimal strips define the line-centers of $R$ and $B$.

We consider several variants induced by orientational constraints on line-centers and provide efficient algorithms for each.
For one line-center, which consists of computing a minimum-width strip that contains at least one point from each pair, we give an $O(n^2)$-time algorithm.
For two line-centers, we obtain an $O(n)$-time algorithm when both are horizontal, and $\Theta(n\log n)$-time algorithms 
when the two centers are parallel or when both orientations are prescribed.
When exactly one orientation is prescribed, we give an $O(n^2)$-time algorithm.
Finally, for the unrestricted case, we present an $O(n^3\log n)$-time algorithm.
\end{abstract}

\section{Introduction}
A \emph{strip} $\sigma$ is the closed region of the plane between two parallel lines, with 
\emph{width} $w(\sigma)$ defined as the perpendicular distance between the two bounding lines.
For a point set $X\subset \mathbb{R}^2$, define the \emph{optimal strip width}
$w^\circ(X)=\min_{\sigma:\, X\subseteq \sigma} w(\sigma)$,
where $\sigma$ may have any orientation.
Let $\sigma^\circ(X)$ denote a minimum-width strip covering $X$.
The \emph{midline} of a strip is the line parallel to its bounding lines and equidistant from them.
The midline of $\sigma^\circ(X)$ is called a \emph{line-center} of $X$.
We refer to the problem of finding the line-center as the \emph{line-center problem}.

For a set $S$ of unordered pairs of points in $\mathbb{R}^2$, a~\emph{bichromatic assignment} 
(or simply an \emph{assignment}) selects exactly one point from each pair to be red and colors the other blue,
yielding the red and blue point sets $R$ and $B$, respectively.
Our objective is to find a bichromatic assignment that minimizes
$\max\{w^\circ(R),\, w^\circ(B)\}$.
We call any bichromatic assignment achieving the minimum an \emph{optimal assignment} for $S$.
The line-centers induced by $\sigma^\circ(R)$ and $\sigma^\circ(B)$ provide representative central axes for $R$ and $B$, respectively.
We refer to this optimization problem as the \emph{bichromatic line-center problem}.

The bichromatic line-center problem is motivated by geometric optimization under \emph{pairwise selection} constraints: for each input pair, we choose which point is colored red (and the other is colored blue).
This models ambiguous, duplicated, or coupled data where we still seek a concise geometric summary per class. Minimum-width strips offer a robust, orientation-free notion of dominant direction and dispersion, and minimizing $\max\{w^\circ(R),\, w^\circ(B)\}$ yields a balanced objective. Algorithmically, the problem is challenging because the assignment interacts nontrivially with the continuous choice of strip orientations.

Applications include extracting two representative axes from paired or uncertain planar observations. In trajectory and shape analysis, pairs may arise from association ambiguity or sensor uncertainty~\cite{Kirubarajan2004}; selecting one point per pair enables two compact directional summaries via 
line-centers. In clustering and visualization, the objective produces two strip-based linear prototypes of bounded thickness, supporting interpretable comparison of directional trends~\cite{Agarwal2003,Qiao2022}. Similar formulations capture assignment tasks that split pairs across two groups while keeping each group geometrically coherent, e.g., bipartite labeling of matched features or paired annotations in image analysis.

\subsection{Previous work}
There has been substantial work on the (monochromatic) line-center problem for point sets in the plane.
Toussaint~\cite{Toussaint1983} gave a $\Theta(n\log n)$-time algorithm for one line-center of $n$ points.
For two line-centers, Jaromczyk and Kowaluk~\cite{Jaromczyk1995} presented an $O(n^2\log^2 n)$-time algorithm.
See~\cite{Chung2026} for additional approximation results.

Several restricted variants have also been studied.
Chung et al.~\cite{Chung2026} gave a $\Theta(n)$-time algorithm for two horizontal line-centers, and Bae~\cite{Bae2020} gave an $O(n^2)$-time algorithm for two parallel line-centers.
Chung et al.~\cite{Chung2026} also obtained an $O(n\log n)$-time algorithm for the mixed variant (one unrestricted, one horizontal).
Ahn and Bae~\cite{Ahn2026} developed an $O(n\log n)$-time algorithm when the orientations are prescribed, and an $O(n^2\alpha(n)\log n)$-time algorithm when the intersection angle is given, where $\alpha(n)$ is the inverse Ackermann function.
Ahn et al.~\cite{Ahn2025} gave an $O(n^2\log n)$-time algorithm for the special case of a fixed right angle.

Fewer results are known for bichromatic point-centers defined on pairs of points, 
which consists of finding a bichromatic assignment that minimizes the maximum radius of two minimum disks covering $R$ and $B$, respectively.
Arkin et al.~\cite{Arkin2015} studied the bichromatic $L_2$- and $L_\infty$-point-center problems for $n$ pairs of points, giving an $O(n^3\log^2 n)$-time algorithm and an $O(n)$-time algorithm, respectively.
They also gave $(1+\varepsilon)$-approximation algorithms
for the bichromatic $L_2$-point-center problem.
Wang and Xue~\cite{Wang2022} improved the running time for the bichromatic $L_2$-point-center problem to $O(n^2\log^2 n)$, and improved the approximation algorithm to $O\!\left(n+\frac{1}{\varepsilon^{3}}\log^{2}\frac{1}{\varepsilon}\right)$.

Arkin et al.~studied the one-point-center problem for point pairs. They gave an $O(n^2\log n)$-time algorithm for the smallest enclosing disk and an $O(n\log^2 n)$-time algorithm for the square case, with an $\Omega(n\log n)$ lower bound in the algebraic decision-tree model.
Using the approach of~\cite{Khanteimouri2017}, one can obtain a $\Theta(n\log n)$-time algorithm for the square case.

\subsection{Our Results}
We study six variants of the bichromatic line-center problem for $n$ unordered pairs of points in the plane.
An overview of the variants and running times is given in Table~\ref{tab:variants}, and examples are shown in Figure~\ref{fig:ex}.
In the one line-center variant (\textsf{1U}), we compute a minimum-width strip of arbitrary orientation that contains at least 
one point from each input pair.
\begin{table}[ht]
\centering
\caption{Six variants of the bichromatic line-center problem for $n$ point pairs in the plane.}
\small
\setlength{\tabcolsep}{6pt}
\begin{tabular}{@{}lll@{}}
\toprule
 & Variant & Time \\
\midrule
\textsf{~1U} & One line-center, unrestricted (Section~\ref{sec:1U})
  & $O(n^2)$\\
\textsf{~2H} & Two horizontal line-centers (Section~\ref{sec:2P})
  & $\Theta(n)$\\
\textsf{~2P} & Two parallel line-centers (Section~\ref{sec:2P})
  & $\Theta(n\log n)$\\
$\mathsf{~2O_2}$ & Two line-centers of prescribed orientations (Section~\ref{sec:2O2})
  & $\Theta(n\log n)$\\
$\mathsf{~2O_1}$ & Two line-centers with one prescribed orientation (Section~\ref{sec:2O1})
  & $O(n^2)$\\
\textsf{~2U} & Two line-centers, unrestricted (Section~\ref{sec:2U})
  & $O(n^3\log n)$\\
\bottomrule
\end{tabular}
\label{tab:variants}
\end{table}

\begin{figure}[t]%
  \centering
\includegraphics[width=\textwidth]{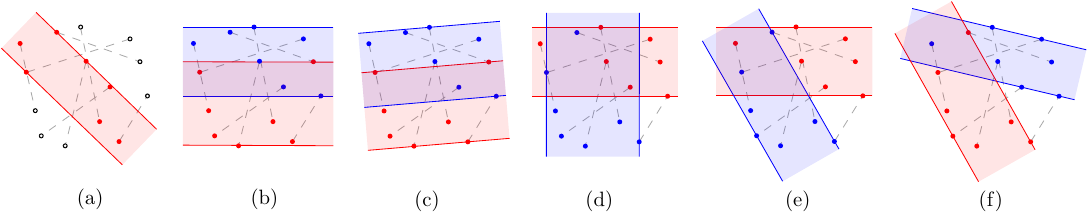}%
  \caption{Optimal assignments for point pairs (marked by dashed segments).
  (a) \textsf{1U}.
  (b) \textsf{2H}.
  (c) \textsf{2P}.
  (d) $\mathsf{2O_2}$ with one horizontal
  and one vertical strips.
  (e) $\mathsf{2O_1}$ with one horizontal strip.
  (f) \textsf{2U}.}
  \label{fig:ex}%
\end{figure}

\section{Preliminaries}
Let $S=\{(p_1,q_1),\ldots,(p_n,q_n)\}$ be a set of $n$ unordered pairs of points in $\mathbb{R}^2$.
For a pair $(p_i,q_i)\in S$, we call $p_i$ and $q_i$ \emph{twins}.
Let $P(S)$ denote the set of points appearing in the pairs of $S$.
For a pair of strips $(\sigma, \sigma')$, define its width as $\max\{w(\sigma), w(\sigma')\}$.
We say that the pair $(\sigma, \sigma')$ is \emph{feasible} with respect to $S$ if
there exists an assignment $(R,B)$ for $S$ such that $R\subset\sigma$ and $B\subset\sigma'$.
When $S$ is clear from the context, we simply say that $(\sigma,\sigma')$ is feasible.
We also say that a strip $\sigma$ is \emph{feasible} if it contains at least one point from each pair in $S$.

\begin{lemma}\label{lem:feasible}
  A pair $(\sigma, \sigma')$ is feasible iff both $\sigma$ and $\sigma'$ are feasible and $P(S)\subset\sigma\cup\sigma'$.
\end{lemma}
\begin{proof}
  The ``Only if'' direction is immediate.   For the other direction, 
  assume $\sigma$ and $\sigma'$ are feasible and $P(S)\subset\sigma\cup\sigma'$.
  For each pair $(p,q)\in S$, if both $p$ and $q$ are contained in $\sigma\cap\sigma'$, assign the two points arbitrarily to different colors.
  Otherwise, assume without loss of generality that $p\in\sigma\setminus\sigma'$; then the feasibility of $\sigma'$ implies 
  $q\in\sigma'$. Assign $p$ red and $q$ blue.
  This yields $(R,B)$ with $R\subset\sigma$ and $B\subset\sigma'$.
  Thus $(\sigma, \sigma')$ is feasible.
\end{proof}

We use the standard Cartesian coordinate system of $\mathbb{R}^2$.
For a point $p$, let $x(p)$ and $y(p)$ be its $x$- and $y$-coordinates, respectively.
A line is \emph{vertical} (resp.,~\emph{horizontal}) if it is parallel to the $y$-axis (resp.,~$x$-axis).
The orientation of a line is $\theta\in[0,\pi)$, measured counterclockwise from the $y$-axis.
We may restrict our attention to strips for which at least one input point lies on each bounding line; these points are called the \emph{defining points} of the strip.

\subparagraph{Duality.}
We use the point--line duality: a point $p=(a,b)$ maps to the line $p^*: y=ax-b$, and a non-vertical line $\ell: y=ax-b$ maps to the point $\ell^*=(a,b)$.
Let $S^*$ denote the set obtained by dualizing each point in every pair of $S$; thus, $S^*$ is a set of unordered line-pairs. For a pair $(p^*_i,q^*_i)\in S^*$,
we call $p^*_i$ and $q^*_i$ \emph{twins}.
A strip $\sigma$ bounded by $y=ax-b_1$ and
$y=ax-b_2$ with $b_1\le b_2$
corresponds to the vertical line segment $\sigma^*=\{(a,b)\mid b\in[b_1,b_2]\}$ in the dual plane.
We call a dual segment $\sigma^*$ \emph{feasible} if it intersects at least one line from each line-pair in $S^*$. By Observation~\ref{obs:dual_strip}, $\sigma^*$ is feasible iff $\sigma$ is feasible.

\begin{observation}\label{obs:dual_strip}
  A strip $\sigma$ contains a point $p$ iff $p^*$ intersects $\sigma^*$. 
\end{observation}

\subparagraph{Arrangement.}
For a set $L$ of $n$ lines in the plane, the arrangement $\arr(L)$ of $L$ is the subdivision of the plane induced by $L$.
A point $p$ has level $k$ with respect to $L$ if exactly $k$ lines of $L$ lie strictly below $p$~\cite[Section~4.7]{Matouvsek2002}.
An edge of $\arr(L)$
is at level $k$ if any relative interior point has level $k$. 
The $k$-th level is the closure of the union of all level-$k$ edges.
Let $L(S^*)$ denote the set of lines appearing in $S^*$.
Since $|L(S^*)|=2n$, $\arr(L(S^*))$ has complexity $O(n^2)$ and can be constructed in
$O(n^2)$ time~\cite{Edelsbrunner1986}.
We denote it by $\arr$ when it is clear from the context.

\section{One line-center problem}\label{sec:1U}
In this section, we consider the one line-center problem (\textsf{1U}).
Given a set $S$ of $n$ point pairs, find a minimum-width strip containing at least one point from each pair, 
i.e., a minimum-width feasible strip.

We solve the \textsf{1U} problem in the dual plane.
After computing $\arr$ in $O(n^2)$ time, we refine it as follows:
for each pair $(\ell_1,\ell_2)\in S^*$, insert on $\ell_1$ a vertex at every $x$-coordinate
of a vertex of $\arr$ lying on $\ell_2$, and symmetrically on $\ell_2$.
Since each line intersects at most $2n-1$ others,
the refinement adds $O(n^2)$ vertices in total and runs in $O(n^2)$ time.
Let $\arrp$ denote the refined arrangement.
For an edge $e$ of $\arrp$ lying on a line $\ell$, let $\lev{e}$ denote its level.
There is an edge on the twin of $\ell$ with the same $x$-range;
call it the \emph{twin edge} of $e$ and denote it by $\twin{e}$.
During the construction of $\arrp$, we compute $\lev{e}$ and $\twin{e}$ for each $e$. 
The whole refinement can be done in $O(n^2)$ time as follows. For a pair $(\ell_1, \ell_2)$, 
we have all vertices of $\arr$ on $\ell_1$ and $\ell_2$ sorted by $x$-coordinates, and the level of edges between two 
consecutive vertices from $\arr$. We sweep two lines from left to right. For each vertex of $\arr$ on $\ell_1$, 
we insert a vertex on $\ell_2$ at the same $x$-coordinate, which splits an edge of $\arr$. We assign $\lev{e}$ and $\twin{e}$ for the new edges in constant time. 
We do the same for each vertex of $\arr$ on $\ell_2$. 

Recall that a strip $\sigma$ in the primal plane corresponds to a vertical line
segment $\sigma^*$ in the dual plane.
For the $i$-th level $\ilev{i}$ of $\arrp$, we consider minimal feasible vertical segments whose lower endpoint lies on $\ilev{i}$.
Let $U_i(x)$ be the minimum possible $y$-coordinate of the upper endpoint of a feasible vertical segment whose lower endpoint 
is the point of $\ilev{i}$ at $x$.
If no such segment exists, define $U_i(x):=\infty$.

To compute a minimum-width feasible strip,
we compute $U_i$ for every $0\le i\le 2n-1$ to find all minimal feasible vertical segments in $\arrp$.
Then we report the one with the minimum width among the dual strips of the feasible vertical segments.
We start with $U_0$. Since the lower endpoint on
$\ilev{0}$ lies below all lines, feasibility depends only on how far upward the segment must extend 
to hit at least one line from each pair. For a line $\ell: y=ax-b$, let $\ell(x)=ax-b$.
Fix $x_0$ and consider a feasible vertical segment
with its lower endpoint on $\ilev{0}$ at $x_0$.
For each pair $(\ell_1,\ell_2)\in S^*$, its upper endpoint must lie on or above $m_Q(x)=\min\{\ell_1(x_0),\ell_2(x_0)\}$.
Hence, $U_0$ is exactly the upper envelope of the functions $m_Q$ over all $Q\in S^*$. 
Since each $m_Q$ is the lower envelope of two lines and consists of two rays,
$U_0$ can be computed in $O(n\log n)$ time~\cite{Hershberger1989}.

We next show how to compute $U_{i+1}$ from $U_i$.
Let $e_i(x)$ be the edge of $\ilev{i}$ at $x$.
Let $\ell_i(x)$ be the line containing $e_i(x)$ and $\twin{\ell_i(x)}$ be the twin of the line.
If $\lev{\twin{e_i(x)}}\le i$, then at $x$ the line $\twin{\ell_i}$ lies on
or below $\ilev{i}$. Hence, for any lower endpoint strictly above $e_i(x)$,
the vertical segment intersects neither $\ell_i$ nor $\twin{\ell_i}$,
and thus it cannot satisfy the pair $(\ell_i,\twin{\ell_i})$.
Thus, $U_{j}(x)=\infty$ for all $j>i$ at this $x$.
Otherwise,
$\twin{\ell_i(x)}$ lies strictly above $e_i(x)$ at $x$.
Then $U_{i+1}(x)=\max\{U_i(x), \twin{\ell_i(x)}(x)\}$.
This observation immediately yields an incremental algorithm. Define
\[
\hat{U}_i(x)=
\begin{cases}
\twin{\ell_i(x)}(x) & \text{if $\lev{\twin{e_i(x)}}>i$} \\
\infty & \text{otherwise.}
\end{cases}
\]
Then $U_{i+1}(x)=\max\{U_i(x),\hat{U}_i(x)\}$.
Since $U_0(x)$ and each $\hat{U}_i(x)$ are piecewise linear, every $U_i(x)$ is piecewise linear as well.
Let $|\cdot|$ denote the number of linear pieces.
During the computation of~$\arrp$, we compute $\twin{e_i(x)}$ and $\lev{\twin{e_i(x)}}$, so
$\hat{U}_i$ can be obtained by tracing $\ilev{i}$ in $O(|\hat{U}_i|)$ time.
Note that $|\hat{U}_i|$ is equivalent to the number of edges in $\ilev{i}$.
Given $U_i$, we can compute $U_{i+1}$ in $O(|U_i|+|\hat{U}_i|)$ time.

For the running time, constructing $\arrp$ and performing the refinement take $O(n^2)$ time, 
and $U_0$ is computed in $O(n\log n)$ time.
For the remaining levels, the total update time is
$
\sum_{i=0}^{2n-1} O\!\bigl(\lvert U_i\rvert + \lvert\hat{U}_i\rvert\bigr)
=
O(n^2)+\sum_{i=0}^{2n-1} O(\lvert U_i\rvert),
$
since the total complexity of all levels in an arrangement of lines is
$O(n^2)$.

For the complexity of $U_i$,
every vertex of the function graph of $U_i$ is a vertex in $\arrp$.
There are $O(n^2)$ vertices in $\arrp$ in total.
Since a single vertex may appear as a vertex of the function graphs of multiple $U_i$'s, a naive bound on
$\sum_{i=0}^{2n-1} O(|U_i|)$
could be as large as $O(n^3)$. We will show that the total complexity
of all curves $U_i$ is in fact only $O(n^2)$.

\begin{lemma}\label{lem:Ui-linear}
$\sum_{i=0}^{2n-1}|U_i|=O(n^2)$.
\end{lemma}
\begin{proof}
We first establish the weaker bound
$\sum_{i=0}^{2n-1} |U_i| = O(n^2\alpha(n))$,
which will later be improved to $O(n^2)$.
Fix a point
$p \in P(S)$ and define $U_p(x)$ as the minimum
possible $y$-coordinate of the upper endpoint of a feasible vertical segment whose lower endpoint is on $p^*$ at $x$.
If none exists, set $U_p(x):=\infty$.
Then $U_p$ is piecewise linear.

To compare $\sum_{i=0}^{2n-1} |U_i|$ with $\sum_{p \in P(S)} |U_p|$, 
refine the graphs as follows: obtain $U_i^+$ from $U_i$ by inserting vertices at all
$x$-coordinates where $\ilev{i}$ has a vertex in $\arrp$, and obtain $U_p^+$
from $U_p$ by inserting vertices at all $x$-coordinates where $p^*$ has a
vertex in $\arrp$.
This adds at most $|\ilev{i}|$ vertices to $U_i$, and
since $\sum_{i=0}^{2n-1} |\ilev{i}| = O(n^2)$, the total number of inserted vertices over all $i$ is $O(n^2)$.
Similarly, each line $p^*$ meets $O(n)$ vertices of $\arrp$, so the total number of inserted vertices over all $p\in P(S)$ is $O(n^2)$.

Fix an edge $e$ of $\ilev{i}$ with $x$-range $[x_1, x_2]$, supported by a line $p^* \in L(S^*)$. 
On this $x$-range, $p^*$ lies at level $i$, and by the
definitions of $U_i$ and $U_p$, we have $U_i(x) = U_p(x)$ for all $x \in [x_1, x_2]$.
By construction, $x_1, x_2$ are vertices of both $U_i^+$ and
$U_p^+$, and neither $\ilev{i}$ nor $p^*$ has any vertex in $x$-range $(x_1, x_2)$.
Thus, the edges of
$U_i^+$ over $[x_1, x_2]$ coincide exactly with the edges of $U_p^+$ over
$[x_1, x_2]$.
Hence, each edge of $U_i^+$ on an $x$-range where
$\ilev{i}$ is supported by $p^*$ corresponds bijectively to the edge of $U_p^+$ on the same
$x$-range, where $p^*$ is at level $i$.
It follows $\sum_{i=0}^{2n-1} |U_i^+| \;=\; \sum_{p \in P(S)} |U_p^+|.$
Since the refinement adds $O(n^2)$ vertices in total,
$\sum_{i=0}^{2n-1} |U_i| + O(n^2) = \sum_{p\in P(S)} |U_p| + O(n^2)$.
It suffices to bound $\sum_{p \in P(S)} |U_p|$.

We bound the complexity of each $U_p$. For each
$(\ell_1,\ell_2)\in S^*$, restrict $\ell_1$ and $\ell_2$ to the portions
lying above $p^*$, which yields two rays. 
A vertical segment
$\sigma^*$ with lower endpoint on $p^*$ intersects $\ell_1$ or $\ell_2$ iff
the upper endpoint of $\sigma^*$ lies above the lower envelope of these two
rays. Thus, $U_p$ is the upper envelope of the lower
envelopes obtained from all pairs.

For each pair, the lower envelope of the two restricted rays consists of
$O(1)$ line segments and rays. Since there are $n$ pairs, the total input
to the upper envelope computation has complexity $O(n)$. The upper envelope of
$O(n)$ line segments and rays has complexity $O(n\alpha(n))$~\cite{Hart1986,Wiernik1988}. Therefore,
$|U_p| = O(n\alpha(n))$ for every $p\in P(S)$. Summing over all lines in $S^*$,
we obtain $\sum_{p\in P(S)} |U_p| = O(n^2\alpha(n))$, and thus
$\sum_{i=0}^{2n-1} |U_i| = O(n^2\alpha(n))$.

We show that $U_p$ has only linear complexity.
The lower envelope for a fixed pair consists of at most two
rays together with $O(1)$ line segments. Moreover, every such line
segment has one endpoint lying on $p^*$. 
We extend it from the endpoint, thereby replacing it with a ray. 
Then the original line segment is exactly the portion of this ray lying above $p^*$. 
Thus, the lower envelope of each pair can be represented by $O(1)$ rays.

Consequently, when we compute the upper envelope over all pairs, the result is
precisely the same as the upper envelope of these rays. Since the upper
envelope of $O(n)$ rays has complexity $O(n)$~\cite{Alegria2024,Sharir1995}, it follows that $U_p$ also has complexity $O(n)$.
Thus, $\sum_{p\in P(S)}|U_p|=O(n^2)$, and hence $\sum_{i=0}^{2n-1}|U_i|=O(n^2)$.
The possible cases of the lower envelope,
depending on the relative position of the pair, are illustrated in
Figure~\ref{fig:pair_lower}.
\end{proof}

\begin{figure}[ht]%
  \centering
\includegraphics[width=\textwidth]{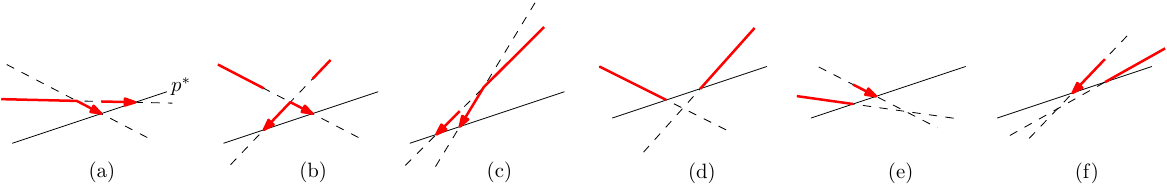}%
  \caption{
  The possible cases of the lower envelope of a pair
$(\ell_1, \ell_2) \in S^*$ (dashed) restricted to the portion above $p^*$,
with arrows showing the extension of each line segment to a ray.
  }
  \label{fig:pair_lower}%
\end{figure}

By Lemma~\ref{lem:Ui-linear}, we can compute $U_i$ for all $0\le i\le 2n-1$ in $O(n^2)$
time. We now compute a minimum-width feasible strip.
Let $\ilev{i}(x)$ be the $y$-coordinate of the point on $\ilev{i}$ at $x$.
If $U_i(x_0)<\infty$ for some $i$ and $x_0$, then there exists a feasible
vertical segment at $x=x_0$ whose lower endpoint lies on $\ilev{i}$ and
whose length is $U_i(x_0)-\ilev{i}(x_0)$. In the primal plane, it corresponds to a feasible strip
of orientation $\theta_0=\tan^{-1}x_0+\pi/2$ and width
$(U_i(x_0)-\ilev{i}(x_0))\cdot\sin\theta_0$.
Thus, for $\theta\in [0,\pi)$  we define the piecewise sinusoidal width function
$w_i(\theta)=(U_i(\tan(\theta-\pi/2))-\ilev{i}(\tan(\theta-\pi/2)))\cdot\sin\theta$.
Breakpoints of $w_i$ occur only when either $U_i$ or $\ilev{i}$ changes its supporting line; hence the complexity of $w_i$ is 
$O(|U_i|+|\ilev{i}|)$, and all $w_i$ can be constructed in overall $O(n^2)$ time.
For each piece, we compute its minimum by checking endpoints and critical points, and take the global minimum
in time linear in the total complexity, yielding a minimum-width feasible strip.

\begin{theorem}\label{thm:one_strip}
    Given $n$ point pairs in $\mathbb{R}^2$, we can solve the \textsf{1U} problem in $O(n^2)$ time.
\end{theorem}

\section{Two parallel line-centers}\label{sec:2P}
We study the two parallel line-center problem, denoted by \textsf{2P}.
Given a set $S$ of $n$ point pairs, find an assignment $\chi=(R,B)$ and two parallel strips covering $R$ and $B$ that minimize the width of the wider strip.
For a fixed $\chi$ and orientation $\theta$, let $\sigma_r(\chi,\theta)$ and $\sigma_b(\chi,\theta)$ denote minimum-width red and blue strips of orientation $\theta$, respectively.

We first consider the special case $\theta=\pi/2$, where both strips are horizontal.
We call it the two horizontal line-center problem, denoted by \textsf{2H}.
For each pair $(p_1,p_2)$ with $y(p_1) \leq y(p_2)$, we say that $p_1$ is the \emph{lower point} and $p_2$ is the \emph{upper point}.

\begin{lemma}\label{lem:horizontal}
  For \textsf{2H}, there exists an optimal assignment
  such that in every input pair, the lower point is colored red and the upper point is colored blue.
  Moreover, such an optimal solution can be computed in $O(n)$ time.
\end{lemma}
\begin{proof}
  Since both strips are horizontal, the width of each strip is determined only by the $y$-coordinates of its defining points.
  Thus, we may project all points onto the $y$-axis, and regard each input pair as a vertical segment connecting its two points.

Let $p_{\min}$ and $p_{\max}$ be the points in $P(S)$ with minimum and maximum $y$-coordinates, respectively.
  Let $\chi^\circ=(R^\circ,B^\circ)$ be the assignment in which, for every input pair, the lower point is colored red and the upper point is colored blue.
  Then $p_{\min}\in R^\circ$ and $p_{\max}\in B^\circ$.
  Suppose, for contradiction, that 
  there exists an assignment $\chi'=(R',B')$ such that
  $\max\{w(\sigma_r(\chi',\pi/2)),w(\sigma_b(\chi',\pi/2))\}<\max\{w(\sigma_r(\chi^\circ,\pi/2)),w(\sigma_b(\chi^\circ,\pi/2))\}$.
  Then $\chi'$ cannot color $p_{\min}$ and $p_{\max}$ the same: otherwise, one strip would contain both and hence have width
$y(p_{\max})-y(p_{\min})$, which is at least the objective value under $\chi^\circ$.
  By swapping the colors if necessary, we may assume $p_{\min}\in R'$ and $p_{\max}\in B'$.

  Since $\chi'\ne\chi^\circ$, there exists a \emph{violating pair} $(p,q)$ with $y(p)<y(q)$, $p\in B'$, and $q\in R'$.
  Let $\chi''$ be obtained from $\chi'$ by swapping the colors of $p$ and $q$.
  For the red strip, $q$ is replaced by $p$, while $p_{\min}$ remains red; thus the maximum $y$-coordinate among red points cannot increase, and
$w(\sigma_r(\chi'',\pi/2))\le w(\sigma_r(\chi',\pi/2))$.
Symmetrically, for the blue strip, $p$ is replaced by $q$, while $p_{\max}$ remains blue; thus the minimum $y$-coordinate among blue points cannot decrease, and
$w(\sigma_b(\chi'',\pi/2))\le w(\sigma_b(\chi',\pi/2))$.
Hence the objective value does not increase.
Repeating this exchange eliminates all violating pairs, transforming $\chi'$ into $\chi^\circ$ without increasing the objective value, 
contradicting the choice of $\chi'$.
Therefore, $\chi^\circ$ is optimal.

To compute $\chi^\circ$, for each pair color the lower point red and the upper point blue, in total $O(n)$ time.
Then compute the red (resp., blue) strip width as $\max_{r\in R^\circ}y(r)-\min_{r\in R^\circ}y(r)$ (resp., $\max_{b\in B^\circ}y(b)-\min_{b\in B^\circ}y(b)$) 
by a single scan in $O(n)$ time. See Figure~\ref{fig:dual_col}(a).
\end{proof}

\begin{figure}[ht]%
  \centering
{\includegraphics[width=.9\textwidth]{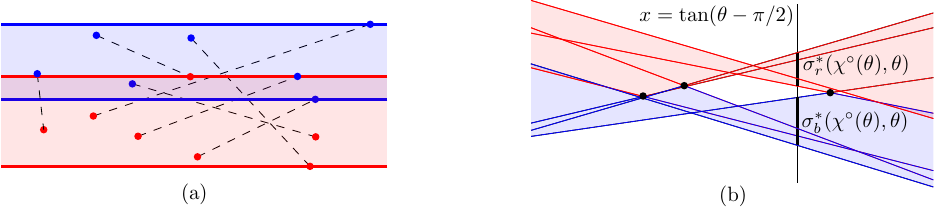}}%
  \caption{
  (a) Optimal assignment at $\theta=\pi/2$.
(b) In the dual plane, $\sigma_r(\chi^\circ(\theta),\theta)$ and $\sigma_b(\chi^\circ(\theta),\theta)$
appear as vertical segments on the line $x=\tan\theta$, which moves from left to
right as $\theta$ increases.}
  \label{fig:dual_col}%
\end{figure}

We maintain an optimal assignment $\chi^\circ(\theta)$ and the
corresponding strips $\sigma_r(\chi^\circ(\theta),\theta)$ and
$\sigma_b(\chi^\circ(\theta),\theta)$ while increasing $\theta$ from $0$ to $\pi$.
We say that a point $p$ is \emph{$\theta$-above} a point $q$ if $p$ lies on or above the line through $q$ with direction $\theta$.
Similarly, we say that $p$ is \emph{$\theta$-below} $q$ if $p$ lies below this line.
For a point set $P$, a point $p\in P$ is \emph{$\theta$-maximal} if it is $\theta$-above every point in $P\setminus\{p\}$, and \emph{$\theta$-minimal} if it is $\theta$-below every point in $P\setminus\{p\}$.
By applying Lemma~\ref{lem:horizontal}, we may assume that, for every input pair, the $\theta$-below point is colored red and the other point is colored blue.
For fixed orientation $\theta$, we choose such an optimal assignment and denote it by $\chi^\circ(\theta)$.
For simplicity, we write $\sigma_r(\theta)$ and $\sigma_b(\theta)$ for $\sigma_r(\chi^\circ(\theta),\theta)$ and $\sigma_b(\chi^\circ(\theta),\theta)$, respectively.

The $\theta$-minimal point of $P(S)$ defines the lower boundary of $\sigma_r(\theta)$, and the $\theta$-maximal point of $P(S)$ defines the upper boundary of $\sigma_b(\theta)$.
To determine the other boundary of $\sigma_r(\theta)$, consider the $\theta$-below point of each input pair.
Among these points, the $\theta$-maximal point defines the upper boundary of $\sigma_r(\theta)$.
Symmetrically, consider the $\theta$-above point of each input pair.
Among these points, the $\theta$-minimal point defines the lower boundary of $\sigma_b(\theta)$.
As $\theta$ increases, it suffices to update the defining points of the two strips, which allows us to maintain their widths.
To do this efficiently, we use the dual transformation $S^*$ of $S$.
In the dual plane, the strips $\sigma_r(\theta)$ and $\sigma_b(\theta)$ correspond to two vertical segments $\sigma_r^*(\theta)$ and $\sigma_b^*(\theta)$ on the vertical line $\ell_\theta:x=\tan(\theta-\pi/2)$.
As $\theta$ increases, $\ell_\theta$ moves from left to right, and we maintain the endpoints of these two vertical segments.
The upper endpoint of $\sigma_r^*(\theta)$ lies on the upper envelope of $\arr$.
The lower endpoint lies on the lower envelope of the upper envelopes defined by the input pairs in $S^*$.
Symmetrically, the lower endpoint of $\sigma_b^*(\theta)$ lies on the lower envelope of $\arr$.
The upper endpoint lies on the upper envelope of the lower envelopes defined by the input pairs in $S^*$.
See Figure~\ref{fig:dual_col}(b).

For each input pair, the lower envelope of its two dual lines consists of two rays.
Hence, over all input pairs, we obtain a collection of $O(n)$ rays.
The upper envelope of these rays has complexity $O(n)$ and can be computed in $O(n\log n)$ time~\cite{Alegria2024,Sharir1995}.
The lower envelope of $\arr$ also has complexity $O(n)$ and can be computed in $O(n\log n)$ time.
Therefore, all orientations at which an endpoint of $\sigma_r^*(\theta)$ changes can be computed and sorted in $O(n\log n)$ time.
By a symmetric argument, all orientations at which an endpoint of $\sigma_b^*(\theta)$ changes can also be computed and sorted in $O(n\log n)$ time.

\begin{lemma}\label{lem:event}
The number of orientations $\theta$ at which the defining points of $\sigma_r^*(\theta)$ or $\sigma_b^*(\theta)$ change is $O(n)$. 
Moreover, all such orientations can be computed and sorted in $O(n\log n)$ time.
\end{lemma}

Define
$f(\theta)=\max\{w(\sigma_r(\theta)),\,w(\sigma_b(\theta))\}$.
Assume that, for every orientation $\theta$, the defining points of $\sigma_r(\theta)$ and $\sigma_b(\theta)$ are known.
Our goal is to find an
orientation $\theta^\circ$ minimizing $f(\theta)$.
In the dual plane, the widths $w(\sigma_r(\theta))$ and $w(\sigma_b(\theta))$ can be written as piecewise-linear functions of
the sweep parameter, each of complexity $O(n)$ by Lemma~\ref{lem:event}.
Their pointwise maximum is the upper envelope of two piecewise-linear functions and thus also has complexity $O(n)$.
Thus, $f(\theta)$ is a piecewise sinusoidal function of complexity $O(n)$.
Once all defining points of $\sigma_r(\theta)$ and $\sigma_b(\theta)$ are available over $\theta\in[0,\pi)$, 
the function $f$ can be evaluated over all intervals in total $O(n)$ time.

After computing and sorting all event orientations, we can generate all candidate orientations for minimizing $f(\theta)$ in additional $O(n)$ time, evaluate $f(\theta)$ at each candidate, and select $\theta^\circ$ attaining the minimum.
Once $\theta^\circ$ is found, we can reconstruct an optimal assignment and the two strips in $O(n)$ time.

\begin{theorem}\label{thm:parallel}
  Given $n$ point pairs in $\mathbb{R}^2$, we can solve the \textsf{2P} problem in $\Theta(n\log n)$ time.
\end{theorem}
\begin{proof}
  By Lemma~\ref{lem:horizontal} and Lemma~\ref{lem:event}, the $O(n)$ candidate orientations for minimizing $f(\theta)$ can be computed in $O(n\log n)$ time.
  After these candidate orientations are computed and sorted, we evaluate $f(\theta)$ at each candidate orientation and select an orientation $\theta^\circ$ minimizing $f(\theta)$.
  Once $\theta^\circ$ is found, the corresponding optimal assignment and the two strips can be reconstructed in $O(n)$ time.
  Therefore, the overall running time is $O(n\log n)$.

  We show that this running time is optimal.
  The lower bound follows by a reduction from the line-center problem.
  Recall that the line-center problem asks, for a given set of points in the plane, to find 
  a line minimizing the maximum distance from the line to the points.
  Equivalently, it asks for a minimum-width strip covering all input points.
  This problem has an $\Omega(n\log n)$ lower bound in the algebraic decision-tree model~\cite{Lee1986}.

Given an instance $P$ of the line-center problem, 
construct a \textsf{2P} instance $S$ by replacing each $p\in P$ with a pair of coincident points at $p$.
  Under any assignment, one copy of $p$ is colored red and the other blue.
  Hence, the red and blue point sets are both exactly $P$.

Therefore any feasible \textsf{2P} solution consists of two parallel strips of the same orientation, one covering the red set and one covering the blue set, i.e., both covering $P$.
For any fixed orientation, each strip must have width at least the width of a minimum strip covering $P$ in that orientation; conversely, any strip covering $P$ can be used for both colors.
Hence the optimum value of the constructed \textsf{2P} instance equals the optimum value of the original line-center instance.
It follows that \textsf{2P} requires $\Omega(n\log n)$ time in the same model.
Combining the upper and lower bounds yields a tight $\Theta(n\log n)$ running time.
\end{proof}

\section{Two line-centers with prescribed orientations}\label{sec:2O2}
We study the $\mathsf{2O_2}$ problem in which the orientations of $\sigma_r$ and $\sigma_b$ are prescribed as $\phi_r$ and $\phi_b$, respectively.
Without loss of generality, assume that $\sigma_b$ is vertical ($\phi_b=0$). If $\phi_r=0$, the problem is solvable in $O(n)$ time by Lemma~\ref{lem:horizontal}, so we assume $\phi_r\neq 0$.

Sort the points in $P(S)$ by their order in the direction orthogonal to $\phi_r$,
from $\phi_r$-below to $\phi_r$-above.
Let $p_1,\ldots,p_{2n}$ be the resulting sequence, and for $1\le i\le j\le 2n$
define $P(i,j)=\{p_i,\ldots,p_j\}$ and $Q(i,j)$ to be $P(S)\setminus P(i,j)$.

Consider any feasible pair $(\sigma_r, \sigma_b)$. 
Then there exist indices $1\le i\le j\le 2n$ such that $\sigma_r\cap P(S)=P(i,j)$. 
Consequently, $Q(i,j)\subset \sigma_b$. By Lemma~\ref{lem:feasible}, both $\sigma_r$ and $\sigma_b$ are feasible.

For each $(i,j)$, let $\sigma_r(i,j)$ be the minimum-width strip of orientation $\phi_r$ covering $P(i,j)$, and let
$\sigma_b(i,j)$ be the minimum-width \emph{feasible} vertical strip of $S$ that covers $Q(i,j)$. 
By the observation above, an optimal solution is attained by some pair $(\sigma_r(i,j),\sigma_b(i,j))$ with $\sigma_r(i,j)$ feasible for $S$.
Although there are $O(n^2)$ index pairs, we will show that it suffices to probe only $O(n)$ pairs to find an optimum.
Let
\begin{equation*}
\begin{aligned}
w_1(i,j) = w(\sigma_r(i,j))\quad \text{ and }\quad
w_2(i,j) =
\begin{cases}
w(\sigma_b(i,j)) & \text{if $\sigma_r(i,j)$ is a feasible strip of $S$,}\\
\infty & \text{otherwise.}
\end{cases}
\end{aligned}
\end{equation*}

Then our goal is to find $(i,j)$ that minimizes $w(i,j)=\max\{w_1(i,j), w_2(i,j)\}$.
Observe that $w_1$ and $w_2$ are monotone, that is
$w_1(i,j) \le w_1(i,j+1)$ and $w_1(i,j) \le w_1(i-1,j)$,
and
$w_2(i,j) \ge w_2(i,j+1)$ and $w_2(i,j) \ge w_2(i-1,j)$.
For each $i$, let $j_i$ be the largest index of $j$ satisfying $w_1(i,j)\le w_2(i,j)$.
By monotonicity, there exists an index pair minimizing $w(i,j)$ with $j\in\{j_i, j_i+1\}$.

We compute $j_i$ for $i=1,\ldots, 2n$ as follows.
First, set $j=1$ and increase $j$ until $w_1(1,j)\le w_2(1,j)$ and $w_1(1,j+1) \ge w_2(1,j+1)$, thereby obtaining $j_1$.
For each $i=2,\ldots, 2n$, starting from $j=j_{i-1}$, repeatedly increment $j$ until either $j=2n$, or $w_1(i,j)\le w_2(i,j)$ and $w_1(i,j+1) \ge w_2(i,j+1)$ hold. Since $j$ only increases over all iterations, the total number of probed index pairs is $O(n)$.
For every index pair $(i,j)$, we can compute $w_1(i,j)$ in $O(1)$ time and test whether $\sigma_r(i,j)$ is a feasible strip for $S$ in $O(1)$ time.
The feasibility test is supported by maintaining, for each pair in $S$, the number of its points covered by $\sigma_r(i,j)$.

To compute $w_2(i,j)$ efficiently, we determine $\sigma_b(i,j)$.
As we increment $i$ or $j$, 
we maintain the points $p_{\min}(i,j)$ and $p_{\max}(i,j)$ with minimum and maximum $x$-coordinates in $Q(i,j)$, respectively.
When $i$ or $j$ increases, exactly one point is inserted into or deleted from $Q(i,j)$; 
hence we can update $p_{\min}(i,j)$ and $p_{\max}(i,j)$ in $O(\log n)$ time.
Since $\sigma_b(i,j)$ covers $Q(i,j)$, it contains both $p_{\min}(i,j)$ and $p_{\max}(i,j)$.
We construct the following data structure.

\begin{lemma}\label{lem:ds}
We can build a data structure of $S$ in $O(n\log n)$ time that, for any two points $p,q\in P(S)$, 
reports in $O(1)$ time a minimum-width feasible vertical strip covering $p$ and $q$.
\end{lemma}
\begin{proof}
Let $\sigma$ be the vertical strip bounded by lines $x=a$ and $x=b$ with $a<b$, and
represent it by $(a,b)$ in a two-dimensional parameter plane.
Fix a pair $(r,s)\in S$ with $x(r)\le x(s)$.
The strip $\sigma$ covers $r$ or $s$ iff either (i) $a\le x(r)$ and $b\ge x(r)$, 
or (ii) $x(r)<a\le x(s)$ and $b\ge x(s)$.
Equivalently, $(a, b)$ lies in $(-\infty,x(r)]\times [x(r),\infty)\cup (x(r),x(s)]\times [x(s),\infty)$
whose lower boundary is a two-step staircase. See Figure~\ref{fig:ds}(a).

Let $R$ be the intersection of the regions over all pairs in $S$.
Then $\sigma$ is feasible iff $(a,b)\in R$.
Each constraint contributes a two-step staircase, so the lower boundary of $R$ is the upper envelope of $n$ such staircases; 
it has complexity $O(n)$ and can be built in $O(n\log n)$ time.

Now take query points $p,q$ with $x(p)\le x(q)$.
The strip $\sigma$ contains $p$ and $q$ iff 
$(a,b)\in (-\infty,x(p)]\times [x(q),\infty)$.
The strip width is $b-a$, so we minimize it over 
$R\cap (-\infty,x(p)]\times [x(q),\infty)$. 
This minimizing point is attained either: (1) on $y=x(q)$ or $x=x(p)$, or 
(2) a convex vertex of $R$. See Figure~\ref{fig:ds}(b).

Case (1) corresponds to the minimum feasible vertical strip whose left bounding line contains $p$ or 
whose right bounding line contains $q$.
Sort the points of $P(S)$ by their $x$-coordinates as $q_1,\ldots,q_{2n}$, and store for each $q_i$
the index of its twin.
Then, sweep left-to-right with two pointers (left boundary at $q_i$, right boundary moving right only) to compute, 
for every $q_i$, the minimum feasible strip and record its width and right boundary.
Since both boundaries move right only, this takes $O(n)$ time after sorting.
A symmetric right-to-left sweep handles strips whose right boundary passes through a given point.
Thus, for query $(p,q)$ we obtain in $O(1)$ time the best strip for Case (1) (taking $\max\{x(q),b\}$ if needed).

For Case (2), we use the range minimum query data structure (RMQ)~\cite{Dov1984}. 
Let $r_1,\ldots,r_k$ be the convex vertices of $R$ in increasing $x$-order, where $r_i=(x_i,y_i)$.
Define $\mathcal{R}[i]=y_i-x_i$ and build an RMQ on $\mathcal{R}$ to answer
$\arg\min_{l\le m\le r}\mathcal{R}[m]$ in $O(1)$ time after $O(k)$ preprocessing~\cite{Dov1984}.
For each $q_j$, store the maximum index $i$ with $x_i\le x(q_j)$ and the minimum index $i'$ with $x(q_j)\le x_{i'}$,
computed in total $O(n)$ time by a linear merge-sweep of sorted lists $(q_j)$ and $(r_i)$.
Given $(p, q)$, we find the minimum index $i$ with $x(p)\ge x_i$ and the maximum index $i'$ with $x(q)\le x_{i'}$ 
from the stored indices in $O(1)$ time.
Then by querying $(i,i')$ to the RMQ, we find the vertex for Case (2) and the corresponding strip.

Finally, comparing the widths from Cases (1) and (2) yields a minimum-width feasible vertical strip covering 
$p$ and $q$ in $O(1)$ time.
\end{proof}

Hence, we can compute $w_2(i,j)$ in $O(1)$ time for each index pair $(i,j)$ using Lemma~\ref{lem:ds}, after we compute $p_{\min}(i,j)$ and $p_{\max}(i,j)$.
Since we test $O(n)$ index pairs,
we can solve the $\mathsf{2O_2}$ problem in $O(n\log n)$ time.
Now we show that this running time is optimal.

\begin{figure}[ht]%
  \centering
\includegraphics[width=.85\textwidth]{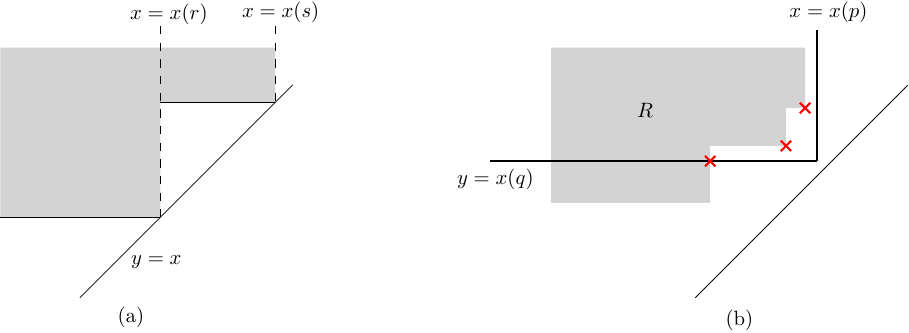}%
  \caption{
(a) The region corresponding to $(r,s)$ with $x(r)\le x(s)$.
(b) Region $R$. Query induced by $[x(p),x(q)]$ and three candidate points, marked by crosses.}
\label{fig:ds}%
\end{figure}

\begin{theorem}\label{thm:2O2}
  Given $n$ point pairs in $\mathbb{R}^2$, we can solve the $\mathsf{2O_2}$ problem in $\Theta(n\log n)$ time.
\end{theorem}
\begin{proof}
  We show the lower bound by reduction from the Max-Gap problem, which 
  consists of finding, for a given set of real numbers, the maximum difference between two consecutive elements in sorted order.
  It has an $\Omega(n\log n)$ lower bound in the algebraic decision-tree model~\cite{Lee1986}.

Consider an instance $X=\{x_1,\ldots,x_n\}$ of Max-Gap with $x_i<x_{i+1}$ for $i=1,\ldots, n-1$.
Let $\Delta=\max_{2\le i\le n}(x_i-x_{i-1})$, and let $L=x_n-x_1$.
For each $x_i\in X\setminus\{x_n\}$, create an input pair $(p_i,q_i)$ with $p_i=(x_i,0)$ and $q_i=(x_i+L,0)$.
For $x_n$, create an input pair $(p_n,q_n)$ with $p_n=(x_n,0)$ and $q_n=(x_n,0)$.
Let $S=\{(p_i,q_i)\mid 1\le i\le n\}$, obtained in $O(n)$ time.

We solve the $\mathsf{2O_2}$ problem on $S$, where the red strip is horizontal and the blue strip is vertical.
Since all points in $P(S)$ lie on the $x$-axis, a red strip of width $0$ can cover all red points for any assignment.
Thus, it reduces to finding a minimum-width feasible vertical strip.

We show that an optimal feasible vertical strip has width $L-\Delta$. 
Let $\sigma_i$ be the minimum-width feasible vertical strip for $S$ whose left boundary passes through $p_i$.
Then $\sigma_1$ must cover all points $p_j$ for $1\le j\le n$, and thus the width of $\sigma_1$ is $L$.
For $i>1$, $\sigma_i$ must cover $q_j$ for every $j<i$ as it cannot cover $p_j$. 
In particular, $\sigma_i$ must cover  $q_{i-1}$.
Thus, $\sigma_i$ has width $L-(x_i-x_{i-1})$.
Since an optimal feasible vertical strip is one with minimum width among $\sigma_1,\ldots,\sigma_n$, 
its width is $L-\Delta$. Thus, the $\mathsf{2O_2}$ problem has an $\Omega(n\log n)$ lower bound.
\end{proof}

\section{Two line-centers with one prescribed orientation}\label{sec:2O1}
We study the two line-center problem for the case that one strip orientation is prescribed, which
we denote by $\mathsf{2O_1}$. 
Without loss of generality, assume that the red strip is prescribed 
to be horizontal. Let $(\sigma_r,\sigma_b)$ be an optimal solution for $S$ such that
$\sigma_r$ is horizontal and there is no solution with
$(\sigma_r,\sigma_b')$ with $w(\sigma_b')<w(\sigma_b)$.

For a feasible pair of strips $(\sigma_1, \sigma_2)$,
we say $\sigma_1$ \emph{captures} $\sigma_2$ if a bounding line $\ell$ of $\sigma_2$ satisfies
$\ell\cap P(S)\subset\sigma_1$. In this case, we denote the bounding line by $\cl(\sigma_2)$.
We consider two cases depending on whether $\sigma_r$ captures $\sigma_b$ or not.

\begin{figure}[ht]%
  \centering
{\includegraphics[width=\textwidth]{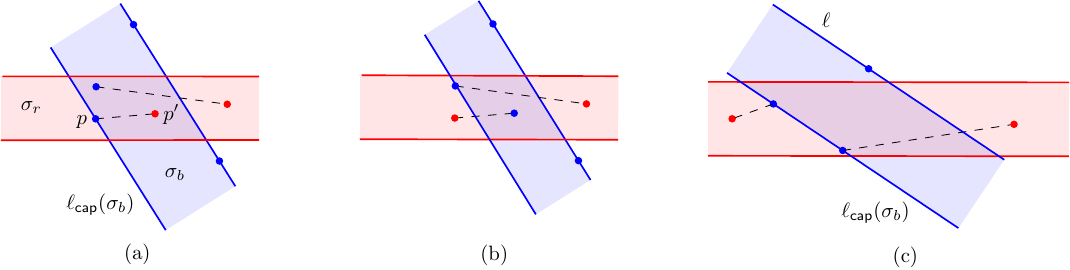}}%
  \caption{
  (a) $\sigma_r$ captures $\sigma_b$. $p$ has its twin $p'$ in the interior of $\sigma_b$.
  (b) A new bichromatic assignment reduces the width of $\sigma_b$.
  (c) $\ell$ contains one point, and $\cl(\sigma_b)$ contains two points whose 
  twins are not contained in the interior of $\sigma_b$.
  }
  \label{fig:captured_prop}%
\end{figure}

\subsection{When \texorpdfstring{$\sigma_r$}{σr} captures \texorpdfstring{$\sigma_b$}{σb}}
Assume that $\sigma_r$ captures $\sigma_b$. We show that there are $O(n^2)$ candidates for $\sigma_b$.
\begin{lemma}\label{lem:capturing_pts}
Suppose $\sigma_r$ captures $\sigma_b$.
Then there exists a point on $\cl(\sigma_b)$ whose twin is not contained in the interior of $\sigma_b$.
Moreover, at least one of the following holds:
\begin{enumerate}
  \item The other bounding line of $\sigma_b$ contains at least two points of $P(S)$.
  \item There is another point on $\cl(\sigma_b)$ whose twin is not contained in the interior of $\sigma_b$.
\end{enumerate}
\end{lemma}
\begin{proof}
Assume that every point in $P(S)\cap \cl(\sigma_b)$ has its twin in the interior of $\sigma_b$.
Let $(R,B)$ be a bichromatic assignment of $S$ such that $R\subset\sigma_r$ and $B\subset\sigma_b$.
Construct $(R',B')$ by swapping the colors of each blue point in $B\cap \cl(\sigma_b)$ and its twin.
Since $\sigma_r$ captures $\sigma_b$, the pair $(\sigma_r,\sigma_b)$ still covers $(R',B')$.

After the swap, no blue point lies on $\cl(\sigma_b)$. Hence we can translate $\cl(\sigma_b)$ inward,
strictly decreasing the width of $\sigma_b$ while keeping all blue points covered, contradicting the minimality of $\sigma_b$.
Therefore, there exists a point $p\in P(S)\cap \cl(\sigma_b)$ whose twin is not contained in the interior of $\sigma_b$.
(See Figure~\ref{fig:captured_prop}(a)--(b).)

Let $\ell$ be the bounding line of $\sigma_b$ other than $\cl(\sigma_b)$.
If $\ell$ contains at least two points of $P(S)$, then (1) holds. Otherwise, $\ell$ contains exactly one point $q\in P(S)$.
Assume for contradiction that $p$ is the only point on $\cl(\sigma_b)$ 
whose twin is not contained in the interior of $\sigma_b$.
By the same swapping argument as above, there is an assignment $(R',B')$ with
$R'\subset\sigma_r$, $B'\subset\sigma_b$, and $p$ is the only blue point on $\cl(\sigma_b)$.

If $q$ is red, then we can translate $\ell$ inward to strictly reduce the width of $\sigma_b$ while keeping $B'\subset\sigma_b$,
a contradiction.
If $q$ is blue, then $\sigma_b$ has exactly one blue point on each bounding line.
In this case, we can rotate $\sigma_b$ while keeping $\cl(\sigma_b)$ passing through $p$ and $\ell$ passing through $q$
so that $B'\subset\sigma_b$ and the width decreases, contradicting minimality~\cite{Houle1988}.
(See Figure~\ref{fig:captured_prop}(c).)
Thus, (2) holds when $\ell$ contains exactly one point of $P(S)$.
\end{proof}

Let $\Sigma$ be the set of feasible strips $\sigma_b$ satisfying the conclusion of Lemma~\ref{lem:capturing_pts}.
\begin{lemma}~\label{lem:capturecompute}
    $|\Sigma|=O(n^2)$, and we can compute $\Sigma$ in $O(n^2)$ time.
\end{lemma}
\begin{proof}
  Consider a strip $\sigma\in\Sigma$.
  Let $p\in \cl(\sigma)$ whose twin $p'$ is not in the interior of $\sigma$, 
  and let $\ell'$ be the other bounding line. 
  Assume that $\cl(\sigma)$ is the lower bounding line of $\sigma$.
  Pick a point $r\in P(S)$ on $\ell'$.
  In the dual plane, the lower endpoint of $\sigma^*$ lies on $r^*$.
  Recall $U_r$ from the proof of Lemma~\ref{lem:Ui-linear}.
  Since $\sigma^*$ is feasible, its upper endpoint lies on or above the function graph of $U_r$.
  Moreover, since $\sigma^*$ does not intersect the dual line of $p'$, and by minimality
  it is the shortest feasible vertical segment with lower endpoint on $r^*$.
  Hence the upper endpoint of $\sigma^*$ lies on the graph of $U_r$.
  
  Consider the case that $\sigma$ satisfies Condition (1) of Lemma~\ref{lem:capturing_pts}.
  Then $\ell'$ contains at least two points of $P(S)$, so the lower endpoint of $\sigma^*$ is a vertex of the arrangement $\arr$.
  Thus, there are $O(n^2)$ such strips.

  Consider the case that $\sigma$ satisfies Condition (2) of Lemma~\ref{lem:capturing_pts}.
  Let $q$ be another point on $\cl(\sigma)$ whose twin is not in the interior of $\sigma$.
  Then the upper endpoint of $\sigma^*$ is determined by the lines $p^*$ and $q^*$ in $L(S^*)$, 
  and the twin lines of $q^*$ and $r^*$ do not intersect $\sigma^*$.
  Therefore, the upper endpoint of $\sigma^*$ is a vertex on the graph of $U_r$.
  Since $U_r$ has complexity $O(n)$, there are $O(n)$ such strips for each $r$, and hence $O(n^2)$ in total.

Finally, we can compute all strips in $\Sigma$ in $O(n^2)$ time using Theorem~\ref{thm:one_strip},
by computing $U_i$ for all levels $i$ (instead of a fixed $U_r$).
The case where $\cl(\sigma)$ is the upper bounding line and $\ell'$ is the lower bounding line follows symmetrically.
\end{proof}

For each $\sigma\in\Sigma$ we need a minimum-width feasible horizontal strip 
that covers $P(S)\setminus\sigma$.
To this end, we first compute the topmost and bottommost points of $P(S)\setminus\sigma$.

\begin{lemma}\label{lem:extreme}
  The topmost and bottommost points in $P(S)\setminus \sigma$
  can be computed in $O(n^2)$ time for all strips $\sigma\in\Sigma$.
\end{lemma}
\begin{proof}
During the construction of the arrangement $\arr$, we store for each edge $e$ two pairs of points:
(i) the topmost and bottommost points of $P(S)$ whose dual lines lie strictly below $e$, and
(ii) the analogous pair for the dual lines strictly above $e$.
We compute the ``below'' information level by level, from $\ilev{0}$ to $\ilev{2n-1}$.
For edges on $\ilev{0}$, no line lies below, so the pair is undefined.
Assume the information is available for all edges of $\ilev{i}$.
Consider an edge $e'$ on $\ilev{i+1}$ and an adjacent edge $e$ on $\ilev{i}$ sharing an endpoint with $e'$.
For any point on $e'$, the set of lines below it equals the set of lines below $e$ plus the line supporting $e$.
Hence, the topmost/bottommost points for the corresponding portion of $e'$ are obtained by comparing the stored pair for $e$
with the point of $P(S)$ dual to the supporting line of $e$, using $O(1)$ comparisons.
Over all edges, this takes $O(n^2)$ time. The ``above'' information is computed symmetrically in $O(n^2)$ time.

Now consider any candidate strip examined between $\ilev{i}$ and $U_i$.
In the dual plane, it corresponds to a vertical segment whose lower and upper endpoints lie on edges 
$e_{\downarrow}$ and $e_{\uparrow}$ of $\arr$, respectively.
The points of $P(S)$ outside the strip correspond exactly to the lines below $e_{\downarrow}$ and the lines above $e_{\uparrow}$.
Therefore, the topmost (resp., bottommost) point in $P(S)\setminus\sigma$ is the better of the two stored topmost (resp., bottommost) candidates for $e_{\downarrow}$ and $e_{\uparrow}$.
Thus, after the $O(n^2)$ preprocessing, the extremes for each such $\sigma$ are obtained in $O(1)$ time.
\end{proof}

By Lemmas~\ref{lem:capturecompute} and~\ref{lem:extreme}, we can compute $\Sigma$ and, 
for every $\sigma\in\Sigma$, the topmost and bottommost points of $P(S)\setminus\sigma$ in $O(n^2)$ total time.
For each $\sigma\in\Sigma$, let $\sigma'$ be a minimum-width feasible horizontal strip covering $P(S)\setminus\sigma$, and consider the candidate pair $(\sigma',\sigma)$.
Since $\sigma'$ must contain the topmost and bottommost points of $P(S)\setminus\sigma$, we can adapt the data structure of Lemma~\ref{lem:ds} to horizontal strips and compute $\sigma'$ in $O(1)$ time per $\sigma$.
Therefore, we generate $O(n^2)$ candidate pairs and select the one of minimum width in $O(n^2)$ total time.

\begin{lemma}\label{lem:2O1c1}
Given $n$ point pairs in $\mathbb{R}^2$, we can compute a minimum-width pair of strips $(\sigma_r,\sigma_b)$ such that $\sigma_r$ is horizontal and $\sigma_r$ captures $\sigma_b$ in $O(n^2)$ time.
\end{lemma}

\subsection{When \texorpdfstring{$\sigma_r$}{σr} does not capture \texorpdfstring{$\sigma_b$}{σb}}\label{sec:2O1c2}
\begin{lemma}\label{lem:noncap}
  Suppose the points of $P(S)$ are sorted by nondecreasing $y$-coordinates.
  For a horizontal line $\ell$, we can find in $O(n)$ time a minimum-width feasible pair of strips $(\sigma,\sigma')$ such that $\ell$ is the lower bounding line of $\sigma$ and $\sigma$ does not capture $\sigma'$, if such a pair exists.
\end{lemma}
\begin{proof}
 If $\sigma$ does not capture $\sigma'$, then each bounding line of $\sigma'$ contains a point of $P(S)\setminus \sigma$.
  Hence both bounding lines of $\sigma'$ are tangent to $\conv{P(S)\setminus \sigma}$.

  For a compact set $X\subset\mathbb{R}^2$, a strip is \emph{supported by} $X$ 
  if it covers $X$ and both its bounding lines are tangent to $X$.
  For any orientation $\theta$, the supported strip is exactly the minimum-width strip covering $X$ in orientation $\theta$.
  Therefore, $\sigma'$ is the supported strip of $\conv{P(S)\setminus \sigma}$ in its orientation.

  We compute $(\sigma,\sigma')$ by sweeping a feasible horizontal strip $\sigma_h$ whose lower bounding line is $\ell$.
  Let $\ell'$ be its upper bounding line, moving upward.
  Let $Q_a$ (resp.\ $Q_b$) be the set of points of $P(S)$ above (resp.\ below) $\sigma_h$.
  While $\ell'$ moves, we maintain $\conv{Q_a\cup Q_b}$.
  Whenever $\ell'$ hits a point of $P(S)$, we update $\conv{Q_a\cup Q_b}$ and recompute the minimum-width feasible supported strip $\sigma_s$ of $\conv{Q_a\cup Q_b}$; if it exists, we report $(\sigma_h,\sigma_s)$.
  The minimum-width reported pair is the desired $(\sigma,\sigma')$.

  \smallskip\noindent
\textbf{Initialization.}
  First compute the minimum-width feasible horizontal strip $\sigma_h$ with lower bounding line $\ell$ in $O(n)$ time by sweeping upward from $\ell$ using the sorted order of $P(S)$; if none exists, report failure.
  Next compute $\conv{Q_a}$ and $\conv{Q_b}$ in $O(n)$ time, and obtain $\conv{Q_a\cup Q_b}$ via the left and right common tangents in $O(\log n)$ time~\cite{Overmars1981}.

  \smallskip\noindent
  \textbf{Finding a minimum-width feasible supported strip.}
  Let $Q^*$ be the set of dual lines of points in $Q_a\cup Q_b$.
  A strip supported by $\conv{Q_a\cup Q_b}$ corresponds in the dual to a vertical segment whose upper (resp.\ lower) endpoint lies on the upper (resp.\ lower) envelope of $Q^*$.
  Thus all supported strips are obtained by tracing such a segment while its endpoints move along the two envelopes.
  During tracing, we maintain for each pair in $S$ how many of its two dual lines are intersected by the segment, and the number of pairs intersected at least once; this tests feasibility.
  The maintained values change only when an endpoint passes a vertex of $\arr$.
  Since each envelope has complexity $O(n)$ and each line of $L(S^*)$ intersects an $x$-monotone concave envelope at most twice, the total number of events on both envelopes is $O(n)$.
  Hence we can enumerate all feasible supported strips and select the minimum-width one $\sigma_s$ in $O(n)$ time.
  We also store the supported strip $\sigma_l$ parallel to the left common tangent (and symmetrically for the right one), together with the faces of $\arr$ containing the endpoints of its dual segment and the per-pair intersection counts obtained during tracing.

  \smallskip\noindent
  \textbf{Updates as $\ell'$ moves.}
  The next point hit by $\ell'$ is found in $O(1)$ time from the sorted order; it is deleted from $Q_a$.
  Because deletions always remove the bottommost point of $Q_a$, we update $\conv{Q_a}$ in amortized $O(1)$ time using~\cite{Wang2025}, and update the common tangents between $\conv{Q_a}$ and $\conv{Q_b}$ in amortized $O(1)$ time using~\cite{Chung2026}.

  If both common tangents remain unchanged, $\conv{Q_a\cup Q_b}$ induces the same supported strip 
  for all orientations, so no new candidate arises.
  Otherwise, either the left tangent rotates clockwise or the right tangent rotates counterclockwise; consider the left case.
  Let the left tangent rotate from $\ell_l$ to $\ell_l'$, with orientations $\theta_l>\theta_l'$.
  Only $\theta\in(\theta_l',\theta_l)$ may change the supported strip, 
  so we trace supported strips as $\theta$ decreases from $\theta_l$ to $\theta_l'$.
  Since $\sigma_{\theta_l}=\sigma_l$ and the dual-location/counters for $\sigma_l^*$ are stored, 
  we continue tracing in $\arr$ from that state.
  Each time an endpoint crosses a vertex of $\arr$, in the primal a bounding line becomes supported by a different point, i.e., 
  an event where a bounding line hits a point of $P(S)$.

  \smallskip\noindent
  \textbf{Event bound.}
  Over the whole sweep, each point of $P(S)$ participates in $O(1)$ such events.
  Indeed, during any traced orientation interval, the left bounding line stays tangent to $\conv{Q_b}$; 
  a point crosses it only at the unique orientation where the left tangent of $\conv{Q_b}$ passes through that point.
  Since traced orientation intervals from different updates are disjoint, each point crosses the left bounding line at most once.
  The right bounding line is tangent to $\conv{Q_a}$ until it passes the orientation of the right common tangent, 
  and afterwards to $\conv{Q_b}$; in each regime, the same argument bounds crossings by one.
  Thus the total number of traced events is $O(n)$, and the whole procedure runs in $O(n)$ time.
\end{proof}

Lemma~\ref{lem:noncap} yields an $O(n^2)$-time algorithm.
Sort $P(S)$ by $y$-coordinates.
For each point $p\in P(S)$, let $\ell$ be the horizontal line through $p$ and 
compute (if it exists) the minimum-width feasible pair $(\sigma,\sigma')$ 
such that $\ell$ is the lower bounding line of $\sigma$ and $\sigma$ does not capture $\sigma'$.
Selecting the minimum-width pair over all $p$ gives the optimal $(\sigma_r,\sigma_b)$ where $\sigma_r$ is horizontal and does not capture $\sigma_b$.
Together with Lemma~\ref{lem:2O1c1}, we obtain the following.

\begin{theorem}\label{thm:2O1}
  Given $n$ point pairs in $\mathbb{R}^2$,
  we can solve the $\mathsf{2O_1}$ problem in $O(n^2)$ time.
\end{theorem}

\section{Two line-centers, unrestricted}\label{sec:2U}
We study the unrestricted two line-center problem, denoted by $\mathsf{2U}$, using the capture relation of Section~\ref{sec:2O1}.
Let $(\sigma_r,\sigma_b)$ be an optimal solution.
As in Section~\ref{sec:2O1}, we assume minimality: there is no feasible $(\sigma_r',\sigma_b)$ with $w(\sigma_r')<w(\sigma_r)$, nor feasible $(\sigma_r,\sigma_b')$ with $w(\sigma_b')<w(\sigma_b)$.

\subsection{When \texorpdfstring{$\sigma_r$}{σr} and \texorpdfstring{$\sigma_b$}{σb} capture each other}
By the same reasoning as in Lemma~\ref{lem:capturing_pts}, both strips satisfy the conditions of Lemma~\ref{lem:capturing_pts}.
By Lemma~\ref{lem:capturecompute}, we can compute in $O(n^2)$ time the set $\Sigma$ of feasible strips satisfying these conditions.
By Lemma~\ref{lem:feasible}, it suffices to find $\sigma,\sigma'\in\Sigma$ with $P(S)\subset \sigma\cup\sigma'$.
Sort $\Sigma$ by nondecreasing width in $O(n^2\log n)$ time, and let $\sigma_1,\ldots,\sigma_{\lvert\Sigma\rvert}$ be the sorted list.
For each $j=1,\ldots,\lvert\Sigma\rvert$, we test whether there exists $i\leq j$ with $P(S)\subset \sigma_i\cup\sigma_j$; if so, we output $(\sigma_i,\sigma_j)$, which is optimal by the ordering.

Let $Q(j):=P(S)\setminus \sigma_j$, and let $Q^*(j)$ be the set of dual lines of points in $Q(j)$.
Then $P(S)\subset \sigma_i\cup\sigma_j$ iff $Q(j)\subset \sigma_i$, which holds iff $\sigma_i^*$ intersects every line in $Q^*(j)$ (Observation~\ref{obs:dual_strip}).
Equivalently, the upper (resp.\ lower) endpoint of $\sigma_i^*$ lies on/above the upper envelope (resp.\ on/below the lower envelope) of $Q^*(j)$.

We maintain a dynamic data structure on $\{\sigma_1,\ldots,\sigma_{j}\}$ that,
given $Q(j)$, report a stored strip covering $Q(j)$, if any; otherwise insert $\sigma_j$.
We use a 2D range tree with fractional cascading~\cite{vanKreveld2008,Chazelle1986} and partial rebuilding~\cite{Overmars1983}, supporting range-maximum queries in $O(\log N)$ time and insertions in amortized $O(\log^2 N)$ time for $N$ weighted points.
For each $\sigma_i\in\Sigma$, let $x_i$ be the $x$-coordinate of $\sigma_i^*$, let $y_i$ be the level of the lower endpoint of $\sigma_i^*$ in $\arr$, and let $\lambda_i$ be the level of the upper endpoint.
Store $\sigma_i$ as the point $(x_i,y_i)$ with weight $\lambda_i$.

For a fixed $j$, compute the upper and lower envelopes $\mathsf{U}_j$ and $\mathsf{L}_j$ of $Q^*(j)$ in $O(n\log n)$ time.
Tracing them in $\arr$ yields $O(n)$ maximal $x$-intervals on which its level is constant; merging the two partitions gives a set $\mathcal I_j$ of $O(n)$ intervals, each annotated with an upper level $u$ (for $\mathsf{U}_j$) and a lower level $l$ (for $\mathsf{L}_j$).
For any interval $[x_1,x_2]\in\mathcal I_j$, a stored strip $\sigma_i$ with $x_i\in[x_1,x_2]$ covers $Q(j)$ iff
$y_i\le l$ and $\lambda_i\ge u$.
Thus, we issue a single range query on
$[x_1,x_2]\times(-\infty,l]$; if the reported maximum weight is at least $u$, we have found a desired $\sigma_i$.
Checking all intervals in $\mathcal I_j$ takes $O(n\log\lvert\Sigma\rvert)=O(n\log n)$ time for this $j$.

We insert at most $\lvert\Sigma\rvert=O(n^2)$ points, for total amortized insertion time $O(n^2\log^2 n)$.
For each $j=1,\ldots,\lvert\Sigma\rvert$, we spend $O(n\log n)$ time on envelope construction and queries, hence $O(n^3\log n)$ time overall, which dominates.

\begin{lemma}\label{lem:2Uc1}
  Given $n$ point pairs in $\mathbb{R}^2$, we can compute a minimum-width pair of strips $(\sigma_r,\sigma_b)$ such that $\sigma_r$ and $\sigma_b$ capture each other in $O(n^3\log n)$ time.
\end{lemma}

\subsection{When \texorpdfstring{$\sigma_r$}{σr} or \texorpdfstring{$\sigma_b$}{σb} does not capture the other}
Without loss of generality, assume that $\sigma_r$ does not capture $\sigma_b$.
By the minimality of $\sigma_r$, some bounding line of $\sigma_r$ contains two points of $P(S)$.
This observation, together with Lemma~\ref{lem:noncap}, yields the following algorithm.

For every choice of two points $p,q\in P(S)$, let $\ell$ be the line through $p$ and $q$ and let $\theta$ be its orientation.
Sort the points in $P(S)$ by their order in the direction orthogonal to $\theta$ in $O(n\log n)$ time.
Then, by Lemma~\ref{lem:noncap}, we can compute in $O(n)$ time a minimum-width feasible pair of strips $(\sigma,\sigma')$, 
if any, such that $\ell$ is a bounding line of $\sigma$ and $\sigma$ does not capture $\sigma'$.
Taking the best solution over all $\binom{n}{2}$ choices of $(p,q)$ gives $(\sigma_r,\sigma_b)$.
The total running time is $O(n^3\log n)$.
Together with Lemma~\ref{lem:2Uc1}, we obtain the following.
\begin{theorem}\label{thm:2U}
    Given $n$ point pairs in $\mathbb{R}^2$, we can solve the \textsf{2U} problem in $O(n^3\log n)$ time.
\end{theorem}

\end{document}